\documentclass[submission,copyright,creativecommons]{eptcs}
\providecommand{\event}{AFL 2023} 

\usepackage[T1]{fontenc}
\usepackage{underscore}

\usepackage{amsthm}
\usepackage{amssymb}
\usepackage{amsmath}
\usepackage{textcomp}
\usepackage{comment}
\usepackage{color}
\usepackage{graphicx}
\usepackage{mathcomp}
\DeclareMathOperator{\mean}{\mathbf{E}}
\DeclareMathOperator{\Prob}{\mathbf{P}}

\newtheorem{definition}{Definition}
\newtheorem{theorem}{Theorem}
\newtheorem{corollary}{Corollary}

\usepackage{longtable}

\usepackage{tikz}
\usepackage{quantikz}
\usepackage{adjustbox}
\newcommand{\cent}[0]{\mbox{\textcent}}
\begin{document}


\title{GAPs for Shallow Implementation of Quantum Finite Automata}

\author{Mansur Ziiatdinov
\institute{Kazan Federal University, Kazan 420008, Russia}
\email{gltronred@gmail.com}
\and
Aliya Khadieva
\institute{Faculty of Computing, University of Latvia, R\={\i}ga, Latvia}
\institute{Kazan Federal University, Kazan 420008, Russia}
\email{aliya.khadi@gmail.com}
\and
Abuzer Yakary{\i}lmaz
\institute{Faculty of Computing, University of Latvia, R\={\i}ga, Latvia}
\email{abuzer.yakaryilmaz@lu.lv}
}
\def\titlerunning{GAPs for Shallow Implementation of QFA}
\def\authorrunning{M. Ziiatdinov, A. Khadieva \& A. Yakary{\i}lmaz}




\maketitle

\begin{abstract}
Quantum fingerprinting is a technique that maps classical input word to a quantum state. The obtained quantum state is much shorter than the original word, and its processing uses less resources, making it useful in quantum algorithms, communication, and cryptography. One of the examples of quantum fingerprinting is quantum automata algorithms for \(MOD_{p}=\{a^{i\cdot p} \mid i \geq 0\}\) languages, where $p$ is a prime number.

  However, implementing such an automaton on the current quantum hardware is not efficient.
  Quantum fingerprinting maps a word \(x \in \{0,1\}^{n}\) of length \(n\) to a state \(\ket{\psi(x)}\) of \(O(\log n)\) qubits, and uses \(O(n)\) unitary operations. Computing quantum fingerprint using all available qubits of the current quantum computers is infeasible due to a large number of quantum operations.

  To make quantum fingerprinting practical, we should optimize the circuit for depth instead of width in contrast to the previous works. We propose explicit methods of quantum fingerprinting based on tools from additive combinatorics, such as generalized arithmetic progressions (GAPs), and prove that these methods provide circuit depth comparable to a probabilistic method. We also compare our method to prior work on explicit quantum fingerprinting methods.


\end{abstract}

\section{Introduction}

A quantum finite state automaton (QFA) is a generalization of classical finite automaton \cite{SY14,AY21}. Here we use the known simplest QFA model \cite{Moore2000}. Formally, a QFA is 5-tuple \( M = (Q,\allowbreak A \cup \{\cent,\mathdollar \},\allowbreak \ket{\psi_{0}},\allowbreak \mathcal{U},\allowbreak \mathcal{H}_{acc})\), where \(Q = \{q_{1}, \ldots, q_{D}\}\) is a finite set of states, \(A\) is the finite input alphabet,  $\cent,\mathdollar$ are the left and right end-markers, respectively. The state of $M$ is represented as a vector \(\ket{\psi} \in \mathcal{H}\), where $\mathcal{H}$ is the $D$-dimensional Hilbert space spanned by $\{ \ket{q_1},\ldots,\ket{q_D} \}$ (here $\ket{q_j}$ is a zero column vector except its $j$-th entry that is 1). The automaton $M$ starts in the initial state \(\ket{\psi_{0}} \in \mathcal{H}\), and makes transitions according to the operators \(\mathcal{U} = \{U_{a} \mid a \in A\}\) of unitary matrices. After reading the whole input word, the final state is observed with respect to the accepting subspace \(\mathcal{H}_{acc} \subseteq \mathcal{H}\).

Quantum fingerprinting provides a method of constructing automata for certain problems. It maps an input word \(w \in \{0,1\}^{n}\) to much shorter quantum state, its fingerprint \(\ket{\psi(w)} = U_{w}\ket{0^{m}}\), where $U_w$ is the single transition matrix representing the multiplication of all transition matrices while reading $w$  and $\ket{0^m} = \underbrace{\ket{0}\otimes \cdots \otimes \ket{0}}_{m~\text{times}} $. Quantum fingerprint captures essential properties of the input word that can be useful for computation.


One example of quantum fingerprinting applications is the QFA algorithms for \(MOD_{p}\) language~\cite{Ambainis2009}. For a given prime number \(p\), the language \(MOD_{p}\) is defined as  \(MOD_{p} = \{ a^{i} \mid i \text{ is divisible by } p\}\). Let us briefly describe the construction of the QFA algorithms for \(MOD_{p}\).

We start with a 2-state QFA $M_k$, where \(k \in \{1,\ldots,p-1\}\). The automaton \(M_{k}\) has two base states \(Q = \{q_{0}, q_{1}\}\), it starts in the state \(\ket{\psi_{0}} = \ket{q_{0}}\), and it has the accepting subspace spanned by \(\ket{q_{0}}\). At each step (for each letter) we perform the rotation
\[
  U_{a} =
  \begin{pmatrix}
    \cos \dfrac{2\pi k}{p} & \sin \dfrac{2\pi k}{p}\\ \\
    -\sin \dfrac{2\pi k}{p} & \cos \dfrac{2\pi k}{p}\\
  \end{pmatrix}.
\]
It is easy to see that this automaton gives the correct answer with probability 1 if \(w \in MOD_p\). However, if $w \notin MOD_p  $, the probability of correct answer can be close to $0$ rather than $1$ (i.e., bounded below by $ 1- \cos^2(\pi/p)  $). To boost the success probability we use \(d\) copies of this automaton, namely \(M_{k_{1}}\), \ldots, \(M_{k_{d}}\), as described below.

The QFA \(M\) for \(MOD_{p}\) has \(2d\) states: \(Q = \{q_{1,0}, q_{1,1}, \ldots, q_{d,0}, d_{d,1}\}\), and it starts in the state \(\ket{\psi_{0}} = \frac{1}{\sqrt{d}} \sum_{i=1}^{d} \ket{q_{i,0}}\). In each step, it applies the transformation defined as:
\begin{align}
  \ket{q_{i,0}} &\mapsto \cos \frac{2\pi k_{i}}{p} \ket{q_{i,0}} + \sin \frac{2\pi k_{i}}{p} \ket{q_{i,1}} \label{eq:u-transform:0}\\
  \ket{q_{i,1}} &\mapsto -\sin \frac{2\pi k_{i}}{p} \ket{q_{i,0}} + \cos \frac{2\pi k_{i}}{p} \ket{q_{i,1}}\label{eq:u-transform:1}
\end{align}
Indeed, $M$ enters into equal superposition of $d$ sub-QFAs, and each sub-QFA applies its rotation.
Thus, quantum fingerprinting technique associates the input word \(w = a^{j}\) with its fingerprint
\[
    \ket{\psi} = \frac{1}{\sqrt{d}} \sum_{i=1}^{d} \cos \frac{2\pi k_{i}j}{p} \ket{q_{i,0}} + \sin \frac{2\pi k_{i}j}{p} \ket{q_{i,1}}.
\]

Ambainis and Nahimovs~\cite{Ambainis2009} proved that this QFA accepts the language \(MOD_{p}\) with error probability that depends on the choice of the coefficients \(k_{i}\)'s. They also showed that for \(d = 2 \log(2p) / \varepsilon\) there is at least one choice of coefficients \(k_{i}\)'s such that error probability is less than \(\varepsilon\). The proof uses a probabilistic method, so these coefficients are not explicit. They also suggest two explicit sequences of coefficients: cyclic sequence \(k_{i} = g^{i} \pmod p\) for primitive root \(g\) modulo \(p\) and more complex AIKPS sequences based on the results of Ajtai et al.~\cite{Ajtai1990}.


While quantum fingerprinting is versatile and has different applications~\cite{Buhrman2001,Ablayev2016a}, it is not practical for the currently available real quantum computers. The main obstacle is that quantum fingerprinting uses an exponential (in the number \(m\) of qubits) circuit depth (e.g., see \cite{KZ22,bsocy2021,salehi2021cost} for some implementations of the aforementioned automaton \(M\)). Therefore, the required quantum volume\footnote{Quantum volume is an exponent of the maximal square circuit size that can be implemented on the quantum computer~\cite{Cross2019QuantumVolume,Wack2021QuantumPerformance}.} \(V_{Q}\) is roughly \(2^{|w| \cdot 2^{m}}\). For example, IBM reports~\cite{IBMEagle} that its Falcon r5 quantum computer has 27 qubits with a quantum volume of 128. It means that we can use only 7 of 27 qubits for the fingerprint technique.



In this paper, we investigate how to obtain better circuit depth by optimizing the coefficients used by $M$: $k_1,\ldots,k_d$.
We use generalized arithmetic progressions for generating a set of coefficients and show that such sets have a circuit depth comparable to the set obtained by the probabilistic method.

We summarize the previous and our results in Table~\ref{tbl:results}. Note that \(p\) is exponential in the number of qubits \(m\). The depth of the circuits is discussed in Section~\ref{sec:shallow}.
\begin{table}
  \caption{Comparison of different methods.}\label{tbl:results}
  \begin{tabular}{|l|l|l|l|p{4.5cm}|}
    \hline
    Method & Width & Depth & Source & Note\\
    \hline
    Cyclic 
           & \(p^{c/\log \log p}\) & \(p^{c/\log \log p}\) & \cite{Ambainis2009} & for some constant \(c > 0\)\\
    AIKPS 
           & \(\log^{2+3\epsilon} p\) & \((1 + 2\epsilon) \log^{1+\epsilon}p\; \log\log p\) & \cite{Razborov1993} & \\
    Probabilistic & \(4\log (2p) / \varepsilon\) & \(2\log (2p) / \varepsilon\) & \cite{Ambainis2009} &\\
    GAPs 
           & \(p / \varepsilon^2\) & \(\lceil \log p - 2 \log \varepsilon \rceil + 2\) & this paper &\\
    \hline
  \end{tabular}
\end{table}

The rest of the paper is organized as follows. In Section~\ref{sec:preliminaries} we give the necessary definitions and results on quantum computation and additive combinatorics to follow the rest of the paper. Section~\ref{sec:shallow} contains the construction of the shallow fingerprinting function and the proof of its correctness. Then, we present certain numerical simulations in  Section~\ref{sec:numeric}.
We conclude the paper with Section~\ref{sec:conclusions} by presenting some open questions and discussions for further research.

\section{Preliminaries}\label{sec:preliminaries}

Let us denote by \(\mathcal{H}^{2}\) two-dimensional Hilbert space, and by \((\mathcal{H}^{2})^{\otimes m}\) \(2^{m}\)-dimensional Hilbert space (i.e.,\ the space of \(m\) qubits). We use bra- and ket-notations for vectors in Hilbert space. For any natural number \(N\), we use \(\mathbb{Z}_{N}\) to denote the cyclic group of order \(N\).

Let us describe in detail how the automaton \(M\) works.
As we outlined in the introduction, the automaton \(M\) has \(2d\) states: \(Q = \{q_{1,0}, q_{1,1}, \ldots, q_{d,0}, d_{d,1}\}\), and it starts in the state \(\ket{\psi_{0}} = \frac{1}{\sqrt{d}} \sum_{i=1}^{d} \ket{q_{i,0}}\). After reading a symbol \(a\), it applies the transformation \(U_{a}\) defined by \eqref{eq:u-transform:0}, \eqref{eq:u-transform:1}:
\begin{align*}
  \ket{q_{i,0}} &\mapsto \cos \frac{2\pi k_{i}}{p} \ket{q_{i,0}} + \sin \frac{2\pi k_{i}}{p} \ket{q_{i,1}} \\
  \ket{q_{i,1}} &\mapsto -\sin \frac{2\pi k_{i}}{p} \ket{q_{i,0}} + \cos \frac{2\pi k_{i}}{p} \ket{q_{i,1}}
\end{align*}
After reading the right endmarker \(\$\), it applies the transformation \(U_{\$}\) defined in such way that \(U_{\$} \ket{\psi_{0}} = \ket{q_{1,0}}\). The automaton measures the final state and accepts the word if the result is \(q_{1,0}\).

So, the quantum state after reading the input word \(w = a^{j}\) is \[\ket{\psi} = \frac{1}{\sqrt{d}} \sum_{i=1}^{d} \cos \frac{2\pi k_{i}j}{p} \ket{q_{i,0}} + \sin \frac{2\pi k_{i}j}{p} \ket{q_{i,1}}.\] If \(j \equiv 0 \pmod p\), then \(\ket{\psi} = \ket{\psi_{0}}\), and \(U_{\$}\) transforms it into accepting state \(\ket{q_{1,0}}\), therefore, in this case, the automaton always accepts. If the input word \(w \notin MOD_{p}\), then the quantum state after reading the right endmarker \(\$\) is
\[
  \ket{\psi'} = \frac{1}{d} \Big( \sum_{i=1}^{d} \cos\frac{2\pi k_{i}j}{p} \Big) \ket{q_{1,0}} + \ldots,
\]
and the error probability is
\[
  P_{e} = \frac{1}{d^{2}} \Big( \sum_{i=1}^{d} \cos \frac{2\pi k_{i}x}{p} \Big)^{2}.
\]

In the rest of the paper, we denote by \(m\) the number of qubits in the quantum fingerprint, by \(d = 2^{m}\) the number of parameters in the set \(K\), by \(p\) the size of domain of the quantum fingerprinting function, and by \(U_{a}(K)\) the transformation defined above, which depends on the set \(K\).

Let us also define a function \(\varepsilon: \mathbb{Z}_{p}^{d} \to \mathbb{R}\) as follows:
\[
  \varepsilon(K) = \max_{x \in \mathbb{Z}_{p}} \bigg( \frac{1}{d^{2}} \Big| \sum_{j=1}^{d} \exp \frac{2\pi i k_{j}x}{p} \Big|^{2} \bigg).
\]
Note that \(P_{e} \le \varepsilon(K)\).

We also use some tools from additive combinatorics. We refer the reader to the textbook by Tao and Vu~\cite{TaoVu2006} for a deeper introduction to additive combinatorics.

An additive set \(A \subseteq Z\) is a finite non-empty subset of \(Z\), an abelian group with group operation \(+\). We refer \(Z\) as the ambient group.

If \(A,B\) are additive sets in \(Z\), we define the sum set \(A + B = \{ a + b \mid a \in A, b \in B\}\). We define additive energy \(E(A,B)\) between \(A,B\) to be
\[
E(A,B) = \bigg| \big\{ (a,b,a',b') \in A \times B \times A \times B \mid a + b = a'+b' \big\} \bigg| .
\]

Let us denote by \(e(\theta) = e^{2\pi i \theta}\), and by \(\xi \cdot x = \xi x / p\) bilinear form from \(\mathbb{Z}_p \times \mathbb{Z}_p\) into \(\mathbb{R} / \mathbb{Z}\). Fourier transform of \(f : \mathbb{Z}_{p} \to \mathbb{Z}_{p}\) is \(\hat{f}(\xi) = \mean_{x \in Z} f(x) \overline{e(\xi \cdot x)}\).

We also denote the characteristic function of the set \(A\) as \(1_{A}\), and we define \(\Prob_Z(A) = \widehat{1_{A}}(0) = |A| / |Z|\).

\begin{definition}[\cite{TaoVu2006}]
  Let \(Z\) be a finite additive group. If \(A \subseteq Z\), we define Fourier bias \(\|A\|_{\mathcal{U}}\) of the set \(A\) to be
  \[
    \| A \|_{\mathcal{U}} = \sup_{\xi \in Z \backslash \{0\}} | \widehat{1_A}(\xi) |
  \]
\end{definition}

There is a connection between the Fourier bias and the additive energy.
\begin{theorem}[\cite{TaoVu2006}]\label{thm:fourier-bias-bound}
  Let \(A\) be an additive set in a finite additive group \(Z\). Then
  \[
    \| A \|^4_{\mathcal{U}} \le \frac{1}{|Z|^3} E(A,A) - \Prob_Z(A)^4 \le \| A \|^2_{\mathcal{U}} \Prob_Z(A)
  \]
\end{theorem}

\begin{definition}[\cite{TaoVu2006}]
  Generalized arithmetic progression (GAP) of dimension \(d\) is a set
  \[
    A = \{ x_0 + n_1 x_1 + \ldots + n_d x_d \mid 0 \le n_1 \le N_1, \cdots, 0 \le n_d \le N_d \},
  \]
  where \(x_0, x_1, \ldots, x_d, N_1, \ldots, N_d \in Z\).
  The size of GAP is a product \(N_1 \cdots N_d\).
  If the size of set $A$, \(|A|\), equals to \(N_1 \cdots N_d\), we say that GAP is proper.
\end{definition}


%

\section{Shallow Fingerprinting}\label{sec:shallow}

Quantum fingerprint can be computed by the quantum circuit given in Figure~\ref{fig:deep-circuit}. The last qubit is rotated by a different angle \(2\pi k_{j} x / q\) in different subspaces enumerated by \(\ket{j}\). Therefore, the circuit depth is \(|K| = t = 2^{m}\). As the set \(K\) is random, it is unlikely that the depth can be less than \(|K|\).

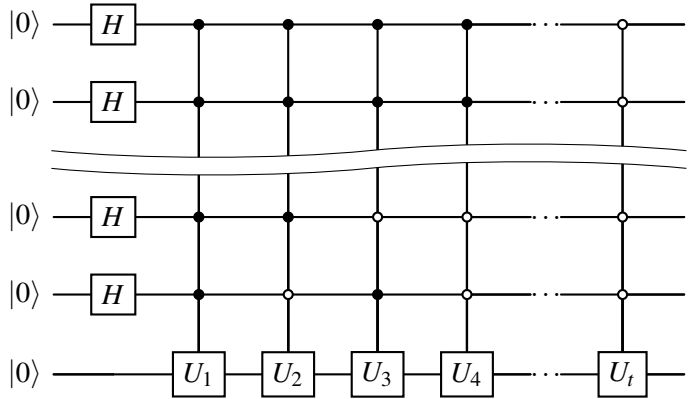
\begin{figure}
  \begin{quantikz}
    \lstick{$\ket{0}$} & \gate{H} & \ctrl{5}   & \ctrl{5}   & \ctrl{5}   & \ctrl{5}   & \qw\ldots & \octrl{5}  & \qw \\
    \lstick{$\ket{0}$} & \gate{H} & \ctrl{4}   & \ctrl{4}   & \ctrl{4}   & \ctrl{4}   & \qw\ldots & \octrl{4}  & \qw \\
    \wave&&&&&&&&& \\
    \lstick{$\ket{0}$} & \gate{H} & \ctrl{2}   & \ctrl{2}   & \octrl{2}  & \octrl{2}  & \qw\ldots & \octrl{2}  & \qw \\
    \lstick{$\ket{0}$} & \gate{H} & \ctrl{1}   & \octrl{1}  & \ctrl{1}   & \octrl{1}  & \qw\ldots & \octrl{1}  & \qw \\
    \lstick{$\ket{0}$} & \qw      & \gate{U_1} & \gate{U_2} & \gate{U_3} & \gate{U_4} & \qw\ldots & \gate{U_t} & \qw
  \end{quantikz}
  \caption{Deep fingerprinting circuit example. Gate \(U_{j}\) is a rotation \(R_{y}(4\pi k_{j}x / p)\). Controls in controlled gates run over all binary strings of length \(s\)} \label{fig:deep-circuit}
\end{figure}

Let us note that fingerprinting is similar to quantum Fourier transform. Quantum Fourier transform computes the following transformation:
\begin{equation}\label{eq:qft}
  \ket{x} \mapsto \frac{1}{N} \sum_{k=0}^{N-1} \omega_{N}^{xk}\ket{k},
\end{equation}
where \(\omega_{N} = e(1/N)\). Here is the quantum fingerprinting transform:
\[
  \ket{x} \mapsto \frac{1}{t} \sum_{j=1}^{t} \omega_{N}^{k_{j}x}\ket{k}.
\]

The depth of the circuit that computes quantum Fourier transform is \(O((\log N)^{2})\), and it heavily relies on the fact that in Eq.~\eqref{eq:qft} the sum runs over all \(k = 0, \ldots, N-1\). Therefore, to construct a shallow fingerprinting circuit we desire to find a set \(K\) with special structure.


Suppose that we construct a coefficient set \(K \subset \mathbb{Z}_{p}\) in the following way. We start with a set \(T = \{t_{1}, \ldots, t_{m}\}\) and construct the set of coefficients as a set of sums of all possible subsets:
\[
  K = \Big\{ \sum_{t \in S} t \mid S \subseteq T \Big\},
\]
where we sum modulo \(p\).


The quantum fingerprinting function with these coefficients can be computed by a circuit of depth \(O(m)\)~\cite{Kalis2018} (see Figure~\ref{fig:shallow-circuit}).

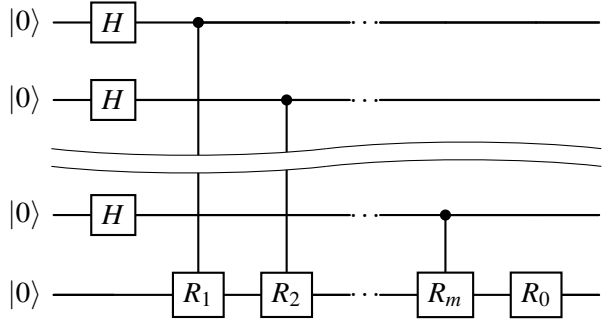
\begin{figure}
  \begin{quantikz}
  \lstick{$\ket{0}$} & \gate{H} & \ctrl{4} & \qw      & \qw\ldots & \qw      & \qw & \qw\\
  \lstick{$\ket{0}$} & \gate{H} & \qw      & \ctrl{3} & \qw\ldots & \qw      & \qw & \qw\\
  \wave&&&&&&& \\
  \lstick{$\ket{0}$} & \gate{H} & \qw      & \qw      & \qw\ldots & \ctrl{1} & \qw & \qw\\
  \lstick{$\ket{0}$} &  \qw     &\gate{R_1}&\gate{R_2}& \qw\ldots &\gate{R_m}& \gate{R_0} & \qw
  \end{quantikz}
  \caption{Shallow fingerprinting circuit example. Gate \(R_{j}\) is a rotation \(R_{y}(4\pi t_{j}x / p)\)} \label{fig:shallow-circuit}
\end{figure}

Finally, let us prove why the construction of the set \(K \subset \mathbb{Z}_{p}\) works.
\begin{theorem}
  Let \(\varepsilon > 0\), let \(m = \lceil \log p - 2 \log \varepsilon \rceil\) and \(d = 2^{m}\).

  Suppose that the number \(t_{0} \in \mathbb{Z}_{p}\) and the set \(T = \{t_{1}, \ldots, t_{m}\} \subset \mathbb{Z}_{p}\) are such that
  \[
    B = \{ 2 t_{0} + n_1 t_1 + \cdots + n_m t_m \mid 0 \le n_1 < 3, \ldots, 0 \le n_m < 3\}
  \]
  is a proper GAP.

  Then the set \(A\) defined as
  \[
    A = \left\{ t_{0} + \sum_{t \in S} t \mid S \subseteq T \right\}
  \]
  has \(\varepsilon(A) \le \varepsilon\).
\end{theorem}

Let us outline the proof of this theorem. Firstly, we estimate the number of solutions to \(a + b = n\). Secondly, we use it to bound the additive energy \(E(A,A)\) of the set \(A\). Thirdly, we bound the Fourier bias \(\|A\|_{{\mathcal{U}}}\). Finally, we get a bound on \(\varepsilon(A)\) in terms of \(p\) and \(m\).

\begin{proof}
  Let us denote a set \(R_{n}(A)\) of solutions to \(a + b = n\), where \(a,b \in A\) and \(n \in \mathbb{Z}_{p}\):
  \[
    R_{n}(A) = \{ (a,b) \mid a+b = n; \; a,b \in A\}.
  \]
  Note that we have \(E(A,A) = \sum_{n\in Z} {R_n(A)}^2\).

  Suppose that \(n\) is represented as \(n = 2t_{0} + \sum_{i=1}^m \gamma_i t_i\), \(\gamma_i \in \{0,1,2\}\). If such representation exists, it is unique, because \(B\) is a proper GAP.
  Let us denote \(c_0 := \{i \mid \gamma_i = 0\}\), \(c_1 := \{i \mid \gamma_i = 1\}\), \(c_2 := \{i \mid \gamma_i = 2\}\).
  It is clear that \(c_0 \uplus c_1 \uplus c_2 = [m]\).

  Now suppose that \(n = a + b\) for some \(a,b \in A\). But \(a = t_{0} + \sum_i \alpha_i t_i\) and \(b = t_{0} + \sum_i \beta_i t_i\), \(\alpha_i, \beta_i \in \{0,1\}\).
  We get that if \(i \in c_{0}\) or \(i \in c_{2}\) then the corresponding coefficients \(\alpha_{i}\) and \(\beta_{i}\) are uniquely determined.
  Consider \(i \in c_{1}\). Then we have two choices: either \(\alpha_{i} = 1; \beta_{i} = 0\), or \(\alpha_{i} = 0; \beta_{i} = 1\).
  Therefore, we have \(R_n(A) = 2^{|c_1(n,A)|}\).

  We have that
  \[
    E(A,A) = \sum_{n \in Z} {R_n(A)}^2 = \sum_{n \in Z} 2^{2|c_1(n,A)|}.
  \]
  Using the fact that \(|c_0(n,A)| + |c_1(n,A)| + |c_2(n,A)| = m\), we see that
  \[
    E(A,A) = \sum_{n \in Z} 2^{2|c_1(n,A)|} = \sum_{j=0}^m \binom{m}{j} 2^{m-j} 2^{2j} = \sum_{j=0}^m \binom{m}{j} 2^{m+j} \le 2^{3m}
  \]

  We can bound the Fourier bias by Theorem~\ref{thm:fourier-bias-bound}:
  \[
    \| A \|^4_{\mathcal{U}} \le \frac{1}{|Z|^3} E(A,A) - \Prob_Z(A)^4 \le \| A \|^2_{\mathcal{U}} \Prob_Z(A)
  \]
  \[
    \| A \|_{\mathcal{U}}^4 \le \frac{2^{3m}}{2^{3 \cdot 2^m}} - \frac{2^{4m}}{2^{4 \cdot 2^m}} = \frac{d^3}{2^{3d}} - \frac{d^4}{2^{4d}}
  \]
  \[
    \| A \|_{\mathcal{U}} \le \frac{d^{3/4}}{p^{3/4}}
  \]

  Finally, we have
  \[
    \varepsilon(A) = \Big( \frac{p}{d} \|A\|_{\mathcal{U}} \Big)^2 \le \frac{p^{1/2}}{d^{1/2}}.
  \]

  By substituting the definitions of \(d\) and \(m\), we prove the theorem.
\end{proof}

%

\begin{corollary}
  The depth of the circuit that computes \(U_{a}(A)\) is \(\lceil \log p - 2 \log \epsilon \rceil\).
\end{corollary}

\begin{theorem}[Circuit depth for AIKPS sequences]
  For given \(\varepsilon > 0\), let
  \begin{align*}
    R &= \{ r \mid r \text{ is prime}, (\log p)^{1+\varepsilon}/2 < r < (\log p)^{1+\varepsilon}\},\\
    S &= \{ 1, 2, \ldots, (\log p)^{1+2\varepsilon}\}, \\
    T &= \{ s \cdot r^{-1} \mid r\in R, s \in S\},
  \end{align*}
  where \(r^{-1}\) is the inverse of \(r\) modulo \(p\).

  Then the depth of the circuit that computes \(U_{a}(T)\) is less than \((1 + 2\epsilon) \log^{1+\epsilon}p\; \log\log p\).
\end{theorem}

\begin{proof}
  Let us denote the elements of \(R\) by \(r_{1}, r_{2}, \ldots\). Let \(S\cdot \{r^{-1}\}\) be a set \(\{s \cdot r^{-1} \mid s \in S\}\).

  Consider the following circuit \(\mathcal{C}_{j}\) (see Figure~\ref{fig:aikps-circuit}) with \(w=\lceil (1+2\varepsilon)\log\log p \rceil + 1\) wires.

  \begin{figure}[h]
    \begin{quantikz}
      \lstick{$\ket{0}$} & \gate{H} & \ctrl{4} & \qw      & \qw\ldots & \qw      & \qw & \qw\\
      \lstick{$\ket{0}$} & \gate{H} & \qw      & \ctrl{3} & \qw\ldots & \qw      & \qw & \qw\\
      \wave&&&&&&& \\
      \lstick{$\ket{0}$} & \gate{H} & \qw      & \qw      & \qw\ldots & \ctrl{1} & \qw & \qw\\
      \lstick{$\ket{0}$} &  \qw     &\gate{R_{j,1}}&\gate{R_{j,2}}& \qw\ldots &\gate{R_{j,w-1}}& \gate{R_{j}} & \qw
    \end{quantikz}
    \caption{Circuit \(\mathcal{C}_j\) for AIKPS subsequence. Gate \(R_{j}\) is a rotation \(R_{y}(4\pi (r_{j}^{-1}) / p)\). Gate \(R_{j,k}\) is a rotation \(R_{y}(2^{k-1} \cdot 4\pi (r_{j}^{-1}) / p)\)} \label{fig:aikps-circuit}
  \end{figure}
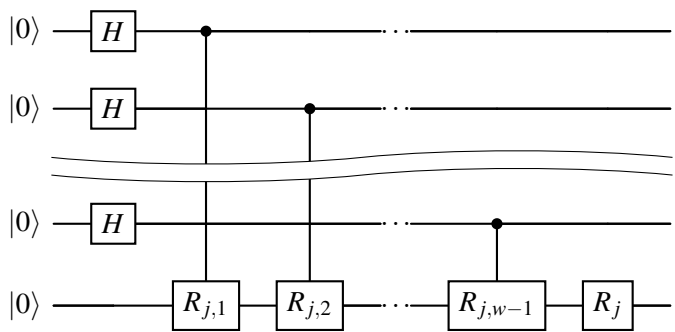

  The circuit \(\mathcal{C}_{j}\) has depth \(\lceil (1+2\varepsilon)\log\log p \rceil + 1\) and computes the transformation \(U_{a}(S \cdot \{r_{j}^{-1}\})\). By repeating the same circuit for all \(r_{j} \in R\) we get the required circuit for \(U_{a}(T)\) (see Figure~\ref{fig:full-aikps-circuit}).

  \begin{figure}[h]
    \begin{adjustbox}{width=\textwidth}
    \begin{quantikz}
\lstick{$\ket{0}$}&\gate{H}& \ctrl{9} & \qw      & \qw\ldots & \qw      &\qw       & \qw      & \qw      & \qw\ldots & \qw      &\qw        & \qw\ldots & \qw            & \qw\\
\lstick{$\ket{0}$}&\gate{H}& \qw      & \ctrl{8} & \qw\ldots & \qw      &\qw       & \qw      & \qw      & \qw\ldots & \qw      &\qw        & \qw\ldots & \qw            & \qw\\
\wave&&&&&&&&&&&&&& \\
\lstick{$\ket{0}$}&\gate{H}& \qw      & \qw      & \qw\ldots & \ctrl{6} &\qw       & \qw      & \qw      & \qw\ldots & \qw      &\qw        & \qw\ldots & \qw            & \qw\\
\lstick{$\ket{0}$}&\gate{H}& \qw      & \qw      & \qw\ldots & \qw      &\qw       & \ctrl{5} & \qw      & \qw\ldots & \qw      &\qw        & \qw\ldots & \qw            & \qw\\
\lstick{$\ket{0}$}&\gate{H}& \qw      & \qw      & \qw\ldots & \qw      &\qw       & \qw      & \ctrl{4} & \qw\ldots & \qw      &\qw        & \qw\ldots & \qw            & \qw\\
\wave&&&&&&&&&&&&&& \\
\lstick{$\ket{0}$}&\gate{H}& \qw      & \qw      & \qw\ldots & \qw      &\qw       & \qw      & \qw      & \qw\ldots & \ctrl{2} &\qw        & \qw\ldots & \qw            & \qw\\
\wave&&&&&&&&&&&&&& \\
\lstick{$\ket{0}$}& \qw    &\gate{R_{1,1}}&\gate{R_{1,2}}& \qw\ldots &\gate{R_{1,w-1}}&\gate{R_1}&\gate{R_{2,1}}&\gate{R_{2,2}}& \qw\ldots &\gate{R_{2,w-1}}&\gate{R_2} & \qw\ldots & \gate{R_{|T|}} & \qw
    \end{quantikz}
  \end{adjustbox}
  \caption{Circuit for \(U_{a}(T)\). Gate \(R_{j}\) is a rotation \(R_{y}(4\pi (r_{j}^{-1}) / p)\). Gate \(R_{j,k}\) is a rotation \(R_{y}(2^{k-1} \cdot 4\pi (r_{j}^{-1}) / p)\)} \label{fig:full-aikps-circuit}
  \end{figure}
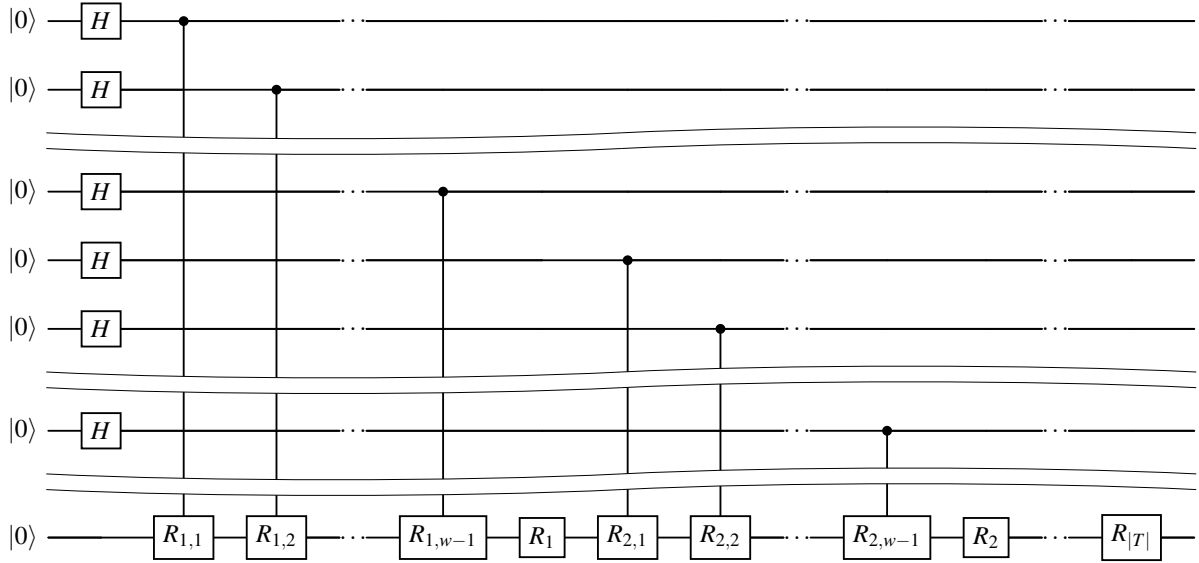

  Since \(|R| < (\log p)^{1+\varepsilon}\), we obtain that the depth of the circuit \(U_{a}(T)\) is less than
  \[
  (1 + 2\epsilon) \log^{1+\epsilon}p\; \log\log p . \qedhere
  \]

\end{proof}

\section{Numerical Experiments}\label{sec:numeric}

We conduct the following numerical experiments. We compute sets of coefficients $K$ for the automaton for the language $MOD_p$  with minimal computational error.

Finding an optimal set of coefficients  is an optimization problem with many parameters, and the running time of a brute force algorithm is large, especially with an increasing number $m$ of control qubits and large values of parameter $p$. Then, the original automaton has $2d$ states, where $d=2^m$. We observe circuits for several $m$ values and use a heuristic method for finding the optimal sets $K$ with respect to an error minimization. For this purpose, the coordinate descend method \cite{wright2015coordinate} is used.

 We find an optimal sets of coefficients for different values of $p$ and $m$ and compare computational errors of original and shallow fingerprinting algorithms for the automaton (see Figure~\ref{fig:s5}). Namely, we set $m=3,4,5$ and find sets using the coordinate descend method  for each case. Even heuristic computing, for $s>5$, takes exponentially more computational time and it is hard to implement on our devices.

\begin{figure}[hp]
\caption{Computational errors for $m=3,4,5$ of original and shallow automata}
\label{fig:s5}
\centering
\includegraphics[width=0.8\textwidth]{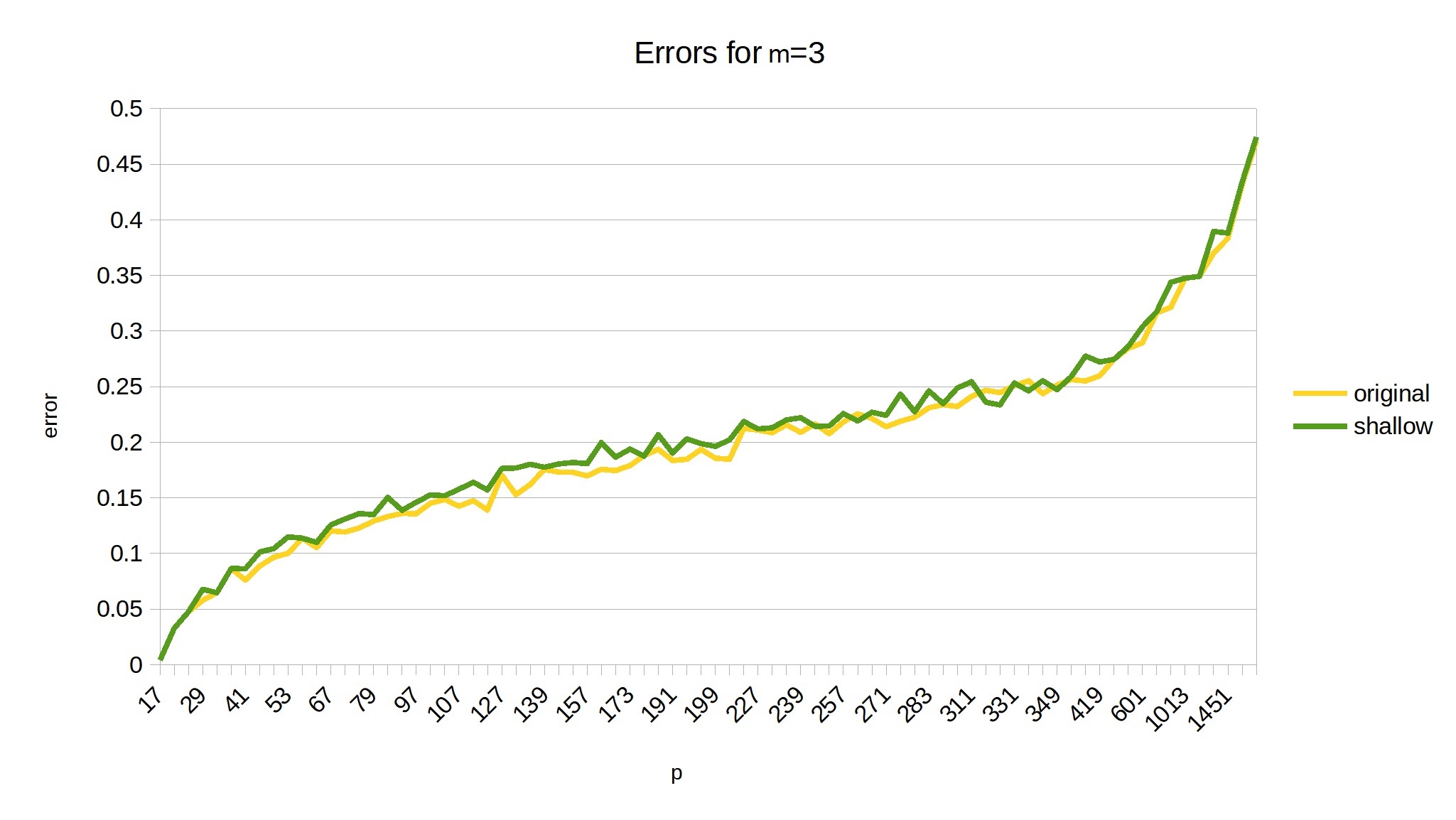}
\includegraphics[width=0.8\textwidth]{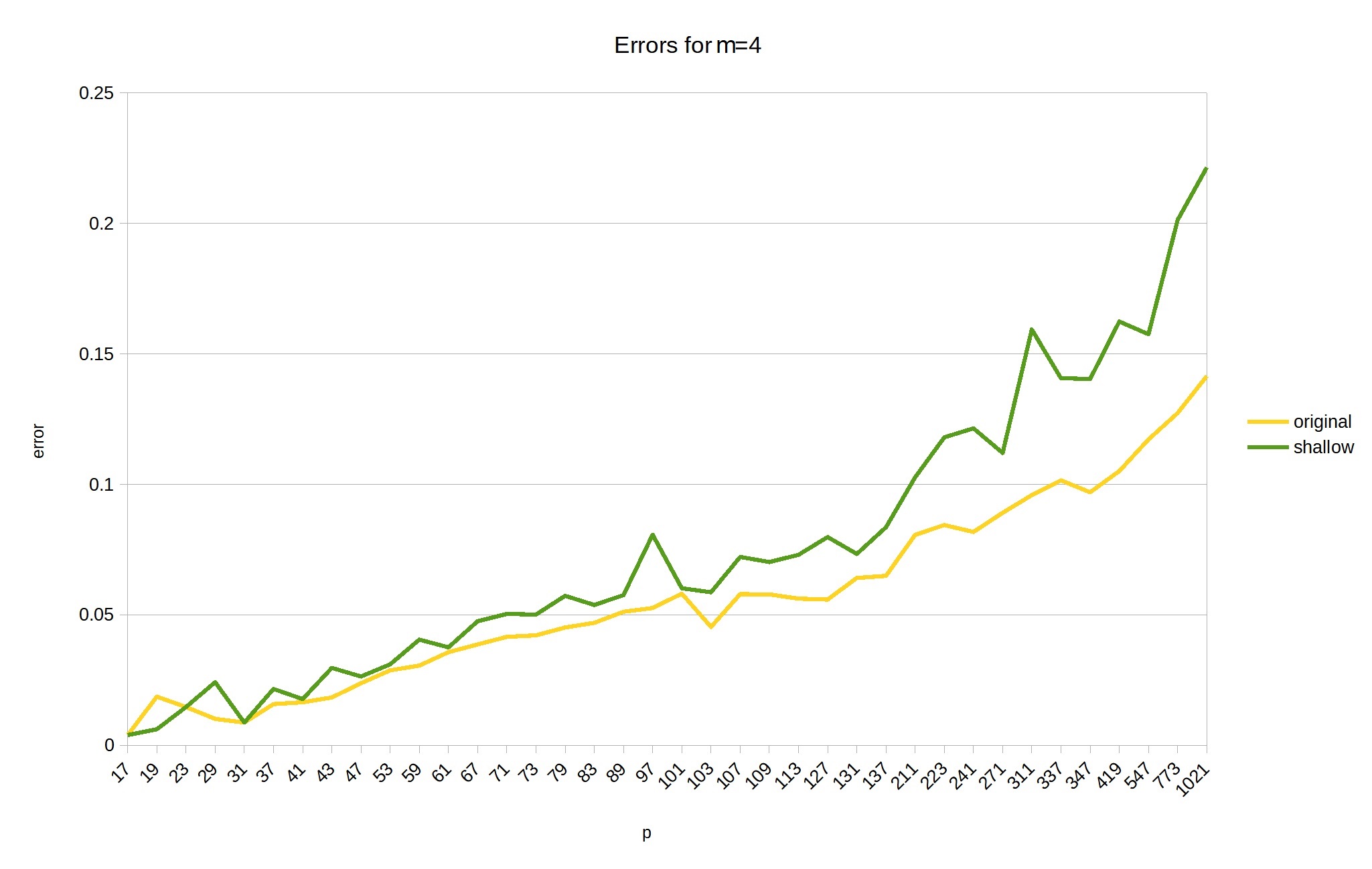}
\includegraphics[width=0.8\textwidth]{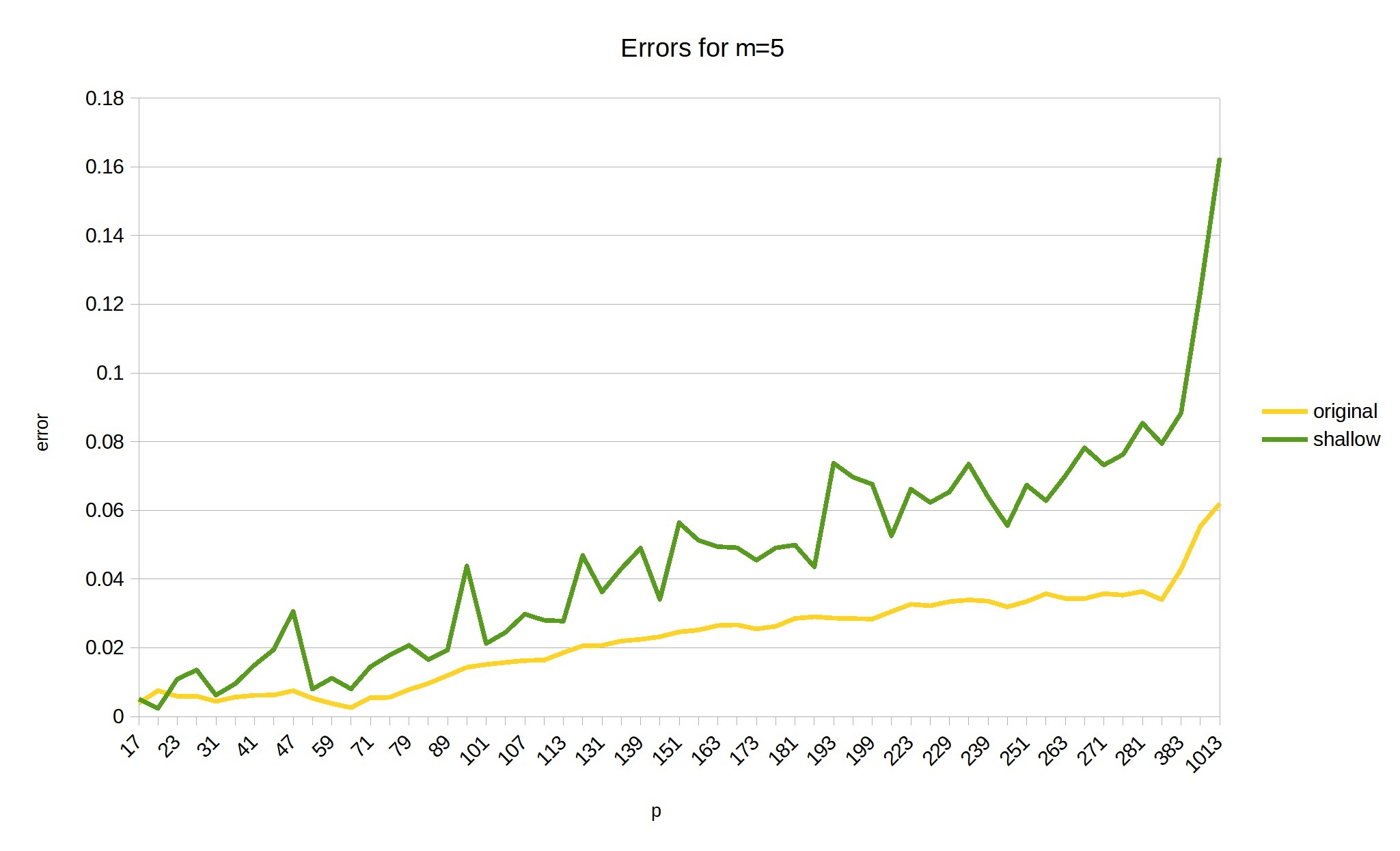}
\end{figure}

One can note that difference between errors becomes bigger with increasing $m$, especially for big values $p$. The program code and numerical data are presented in a git repo \cite{paramsComputing}.

The graphics in Figure~\ref{fig:s5-prop} show a proportion of the errors of the original automaton over the errors of the shallow automaton for $m=3,4,5$ and the prime numbers until 1013.

\begin{figure}[hp]
\caption{Proportions of the shallow automaton errors over the original automaton errors for $m=3,4,5$ and different values of $p$}
\label{fig:s5-prop}
\centering
\includegraphics[width=0.8\textwidth]{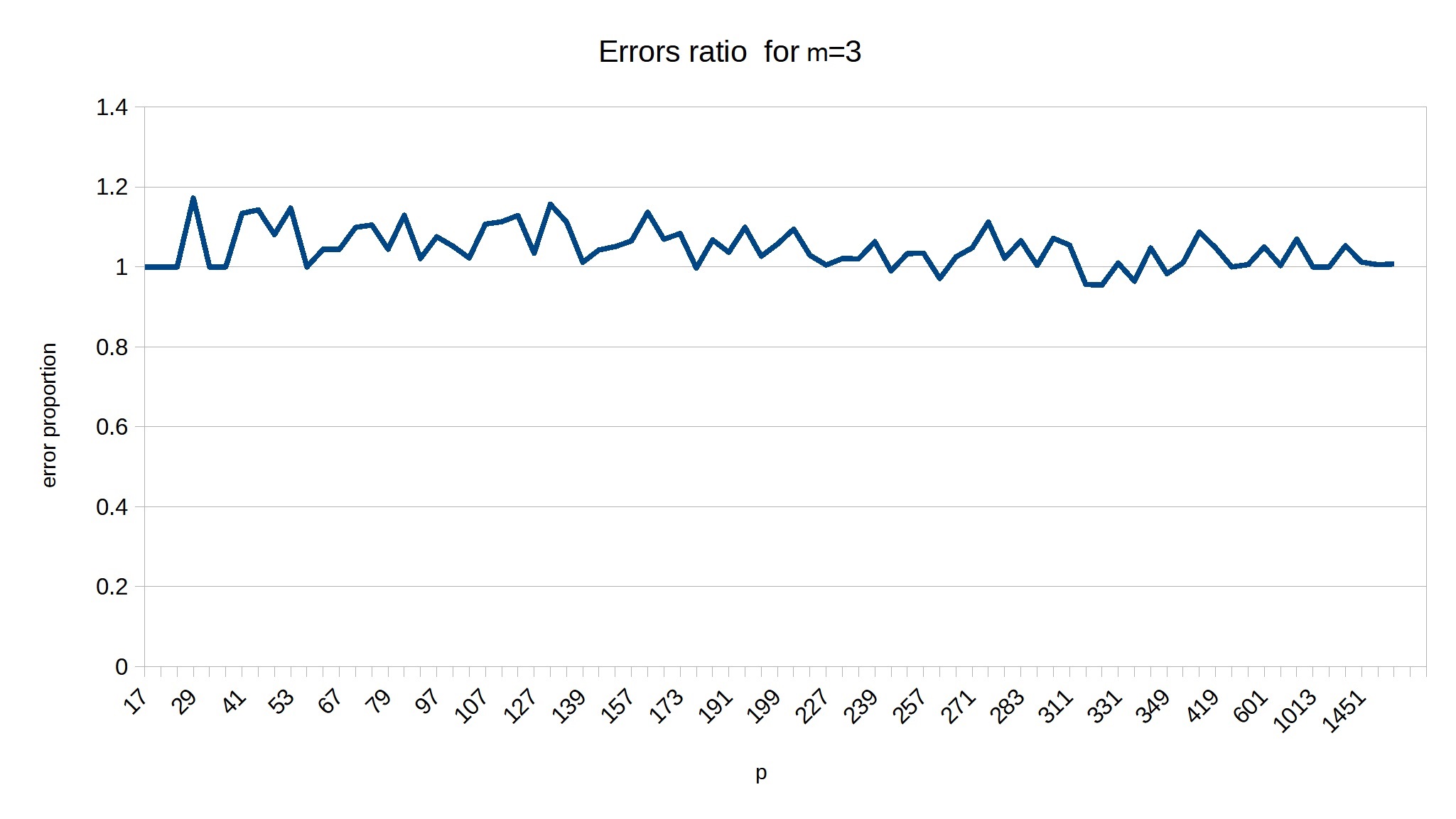}
\includegraphics[width=0.8\textwidth]{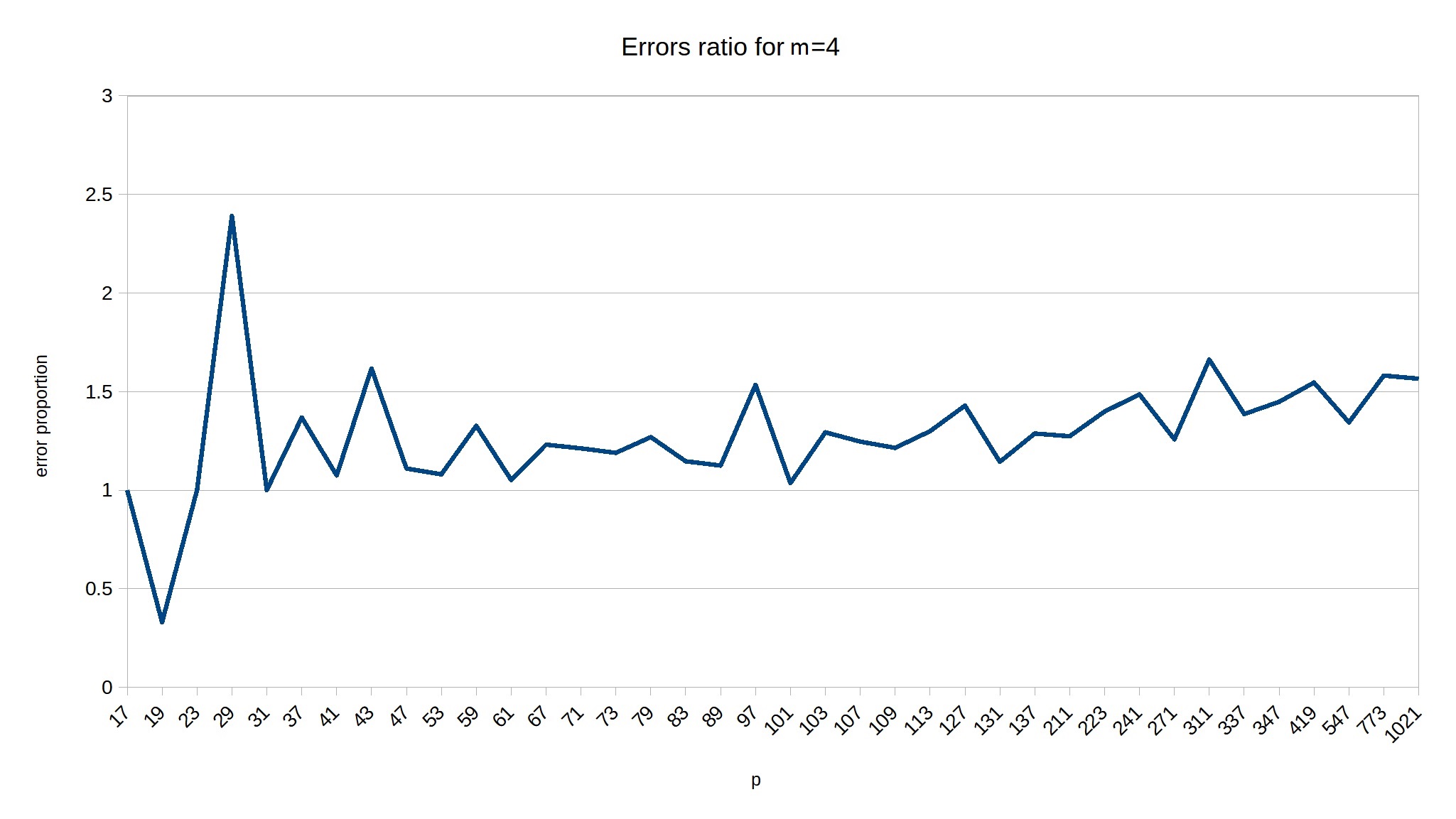}
\includegraphics[width=0.8\textwidth]{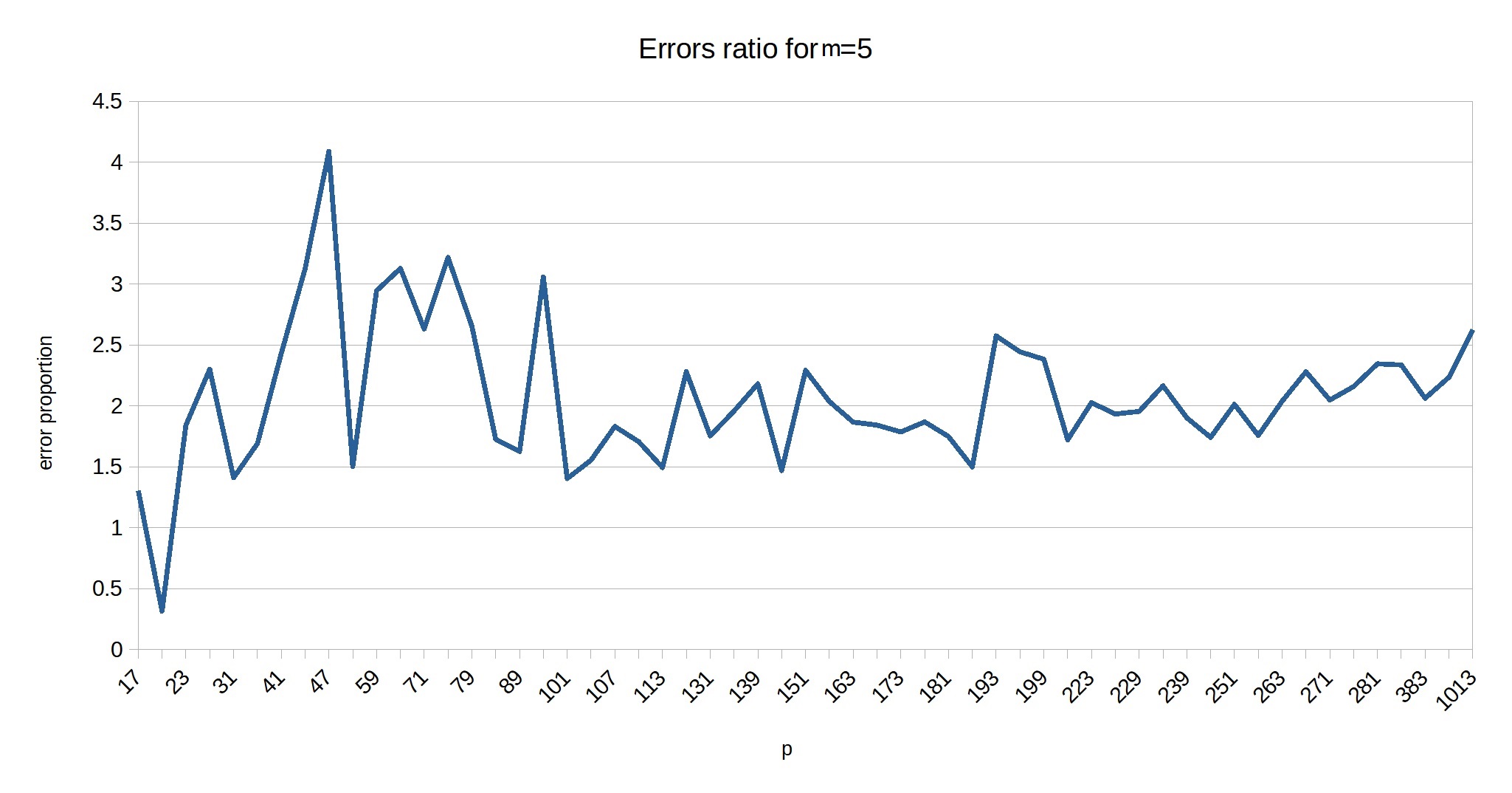}
\end{figure}
 
As we see, for a number of control qubits $m=3$, the difference between the original and shallow automata errors is approximately constant. The ratio of values fluctuates between 1 and 1.2. In the case $m=4$, this ratio is approximately 1.5 for almost all observed values $p$. The ratio of errors is nearly  between 1.5 and 3, for $m=5$.

According to the results of our experiments, the circuit depth $m+1$ is enough for valid computations, while the original circuit uses $O(2^m)$ gates. Since the shallow circuit is much simpler than the original one, its implementation on real quantum machines is much easier. For instance, in such machines as IBMQ Manila or Baidu quantum computer, a ``quantum computer'' is represented by a linearly related sequence of qubits. CX-gates can be applied only to the neighbor qubits. For such a linear structure of qubits, the shallow circuit can be implemented using $3m+3$ CX-gates. Whereas a nearest-neighbor decomposition \cite{mottonen2006decompositions} of the  original circuit requires $O(d \log d)=O(m2^m)$ CX-gates.

\section{Conclusions}\label{sec:conclusions}

We show that generalized arithmetic progressions generate some sets of coefficients \(k_{i}\) for the quantum fingerprinting technique with provable characteristics. These sets have large sizes, however, their depth is small and comparable to the depth of sets obtained by the probabilistic method. These sets can be used in the implementations of quantum finite automata suitable for running on the current quantum hardware.

We run numerical simulations. They show that the actual performance of the coefficients found by our method for quantum finite automata is not much worse than the performance of the other methods.

Optimizing quantum finite automata implementation for depth also poses an open question. The lower bound for the size of \(K\) in terms of \(p\) and \(\varepsilon\) is known~\cite{Ablayev2016a}. Therefore, for given \(p\) and \(\varepsilon\), quantum finite automata cannot have less than \(O(\log p / \varepsilon)\) states. But, to our knowledge, a lower bound for the circuit depth of the transition function implementation is not known. So, we pose an open question: is it possible to implement a transition function with depth less than \(O(\log p)\)? What is the lower bound for it?

\section{Acknowledgments}

Yakary{\i}lmaz was partially supported by the ERDF project Nr. 1.1.1.5/19/A/005 ``Quantum computers with constant memory'' and the project ``Quantum algorithms: from complexity theory to experiment'' funded under ERDF programme 1.1.1.5.

This paper has been supported by the Kazan Federal University Strategic Academic Leadership Program ("PRIORITY-2030").
Research in Section 4 were supported by the subsidy allocated to Kazan Federal University for the state assignment in the sphere of scientific activities, project No. 0671-2020-0065.

\bibliography{shallow-fingerprinting.bib}
\bibliographystyle{eptcs}

\end{document}